\documentclass[prr,aps,10pt,twocolumn,superscriptaddress,floatfix,longbibliography]{revtex4-2}
\usepackage{amsmath,amssymb,amsfonts,amsthm}
\usepackage{graphicx}
\usepackage{enumitem}
\usepackage{bm}
\usepackage{rotating}
\usepackage[colorlinks=true,linkcolor=blue, citecolor=blue, urlcolor=blue, bookmarks]{hyperref}
\usepackage{pbox}
\usepackage{array}
\usepackage{mathtools}
\usepackage{leftidx}
\usepackage{lmodern}
\usepackage[normalem]{ulem}
\usepackage{braket}
\usepackage{tensor}
\usepackage[usenames,dvipsnames]{xcolor}
\usepackage{float}
\usepackage{footnote}
\usepackage{dsfont}
\usepackage{mathrsfs}
\usepackage{bbm}
\usepackage{tikz}
\usepackage{comment}

\usetikzlibrary{calc,matrix,positioning,shapes}

\newcommand{\R}{\mathbb{R}}
\newcommand{\C}{\mathbb{C}}

\newcommand{\M}{\mathcal{M}}

\newcommand{\dd}{\mathrm{d}}
\newcommand{\rr}[1]{\left(#1\right)}

\newcommand{\fock}{{\mathfrak{F}(\mathcal{H})}}

\DeclareMathOperator{\Tr}{Tr}
\renewcommand{\tilde}{\widetilde}

\renewcommand{\openone}{\mathbbm{1}}
\newcommand{\cketbra}[2]{{\left| {#1} \right) \!\!\left( {#2} \right|}}
\newcommand{\cbraket}[1]{{\left( {#1} \right)}}
\newcommand{\cket}[1]{{\left|{#1} \right)}}

\newcommand{\ketbra}[2]{{\left| {#1} \right\rangle \!\!\left\langle {#2} \right|}}
\DeclareMathOperator{\Ad}{Ad}

\DeclareMathOperator{\Span}{span}
\DeclareMathOperator{\Id}{Id}
\newtheorem{proposition}{Proposition}

\newtheorem{lemma}{Lemma}
\newtheorem{definition}{Definition}

\theoremstyle{definition}
\newtheorem{example}{Example}
\theoremstyle{definition}

\newcommand{\exampleqed}{\hfill{$\blacklozenge$}}

\makeatother

\begin{document}

\title{Continuous matrix product operators for quantum fields}

\author{Erickson Tjoa}
\affiliation{Max-Planck-Institut f\"ur Quantenoptik, Hans-Kopfermann-Stra\ss e 1, D-85748 Garching, Germany}
\affiliation{Munich Center for Quantum Science and Technology (MCQST), Schellingstraße 4, D-80799 Munich, Germany}

\author{J. Ignacio Cirac}
\affiliation{Max-Planck-Institut f\"ur Quantenoptik, Hans-Kopfermann-Stra\ss e 1, D-85748 Garching, Germany}
\affiliation{Munich Center for Quantum Science and Technology (MCQST), Schellingstraße 4, D-80799 Munich, Germany}

\begin{abstract}

In this work we introduce an ansatz for continuous matrix product operators for quantum field theory. We show that (i) they admit a closed-form expression in terms of finite number of matrix-valued functions without reference to any lattice parameter; (ii) they are obtained as a suitable continuum limit of matrix product operators; (iii) they preserve the entanglement area law directly in the continuum, and in particular they map continuous matrix product states (cMPS) to another cMPS. As an application, we use this ansatz to construct several families of continuous matrix product unitaries beyond quantum cellular automata.

\end{abstract}

\maketitle

\paragraph*{Introduction.}
Tensor networks have been very useful in providing a compact description of quantum many-body physics with local interactions. Many analytical insights in tensor network theory were also made possible by its compatibility with tools from quantum information theory \cite{cirac2021review}. Tensor networks can be used to describe physically relevant states known as matrix product states (MPS) \cite{Hastings_2007,perez2006matrix}, and also to describe operators known as matrix product operators (MPO) \cite{verstraete2004mpdo,zwolak2004mpdo}. The latter is particularly useful to describe mixed states \cite{verstraete2004mpdo,zwolak2004mpdo,Cuevas2013mpdo,verstraete2018convex,liu2025parent,kato2024exact}, symmetries \cite{GarreRubio2023classifyingphases,lootens2021mposymm,bultinck2017mpoa,molnar2022matrix} and approximate short-time evolutions \cite{Pirvu_2010mpo-evol,zaletel2015longrangeCMPO}. 

Tensor networks are designed to capture the notion of area-law entanglement in quantum many-body systems \cite{wolf2007arealaw,piroli2020qca,orus2014,eisert2010arealaw}. This feature allows low-energy states of local Hamiltonians to be well-approximated by MPS \cite{Hastings_2007,brandao2015exponential,verstraete2006faithful,schuch2008entropy}. Similarly, the fact that quantum cellular automata (QCA) \cite{schumacher2004reversible,gross2012index,Farrelly2020reviewofquantum} (unitaries with exact light cone) can only create little entanglement is captured by its equivalent representation as translationally-invariant {matrix product unitaries} (MPU) \cite{cirac2017mpu,xie2018mpu}. However, formulating a good notion of area-law entanglement is non-trivial in the continuum, i.e., quantum fields. This is relevant whenever one wishes to preserve some properties that only exist in the continuum (e.g., symmetries) or when the system of interest is naturally defined as a quantum field.

The first attempt towards porting tensor network techniques directly to the continuum was the construction of {continuous matrix product states} (cMPS) \cite{verstraete2010cmps,haegeman2013cmps,osborne2010holographic}. They have been successful as a variational ansatz for strongly interacting (non-)relativistic quantum systems \cite{guifre2015liebliniger,ganahl2017liebliniger,tilloy2021relativisticCMPS,tilloy2022study,tuybens2022variationalCMPS} and as a continuum limit of MPS, they capture the entanglement area law directly in the continuum. Furthermore, cMPS can be viewed as a real-space renormalization procedure applied to one-dimensional MPS run backwards, i.e., performing fine-graining procedure rather than coarse-graining by blocking \cite{gemma2018continuum,gemma2020generalizedMPS}.

In this work we provide an ansatz for \textit{continuous matrix product operator} (cMPO) for quantum fields \footnote{We note that in \cite{tang2020cMPOfinitetemp} a notion of cMPO was introduced to capture the continuous (imaginary) time evolution of a lattice system in terms of tensor networks. Our ansatz is consistent with their definition if one exchanges time by space. Furthermore, by expressing the MPOs in the Fock-like representation, it allows us to take the continuum limit directly and explicitly express it in terms of bosonic operators, in much the same way cMPS are \cite{verstraete2010cmps}.}. We formulate cMPO as a suitable continuum limit of discrete MPO without reference to any underlying lattice spacing. We show that it admits a closed-form expression in terms of a path-ordered exponential similar to traced Wilson line operator in non-Abelian gauge theory. By construction it preserves area-law entanglement natively in the continuum and they map a cMPS to another cMPS. As an application of this ansatz, we are able to construct continuous matrix product unitaries (cMPU) that are  continuum limit of MPUs beyond QCA \cite{Styliaris2025matrixproduct}.  

\paragraph*{Review of cMPS.}
Let us first review the basic features of cMPS. Our setting is a non-relativistic bosonic field defined on a one-dimensional interval $I \subset \R$ of length $\ell$. In the second quantization formalism, a bosonic field operator $\psi(x)$ satisfies the canonical commutation relations $[\psi(x),\psi^\dagger(y)] = \delta(x-y)\openone$ and a vacuum state $\ket{\Omega}$ with respect to $\psi$ is defined to be the state for which $\psi(x) \ket{\Omega} = 0$ for all $x$. The Hilbert space of the field is a Fock space $\fock \coloneqq \bigoplus_{N=0}^\infty \mathcal{H}^{\odot N}$ where $\mathcal{H}\cong L^2(I)$ is the one-particle sector Hilbert space and  $\mathcal{H}^{\odot N}$ is the symmetrized subspace of $N$ identical particles. By construction any arbitrary state in $\fock$ must have finite total particle number. 

\begin{definition}[\cite{verstraete2010cmps,haegeman2013cmps}]
    \label{definition: cMPS}
    A bosonic cMPS with bond dimensions $D$ is the state $\ket{\psi[B,Q,L]}\in \fock$ given by
    \begin{align}
        \hspace{-0.085cm}\ket{\psi[B,Q,L]}
        &\coloneqq \Tr_D(B\mathcal{P}e^{ \int_I \dd x\,Q(x)\otimes \openone + L(x)\otimes \psi^\dagger(x)})\ket{\Omega}
        \label{eq: cMPS-def}
    \end{align}
    where $\mathcal{P}$ denotes path-ordering, $\ket{\Omega}\in \fock$ is the Fock vacuum state annihilated by $\psi(x)$, $Q(x) $ and $ L(x)$ are one-parameter family of matrices \footnote{Here we use $L(x)$ rather than $R(x)$ as is common in the cMPS literature \cite{haegeman2013cmps,verstraete2010cmps}, as we will associate $L(x),R(x)$ with left- and right-multiplication maps respectively when defining cMPO.} in $M_D(\C)$, $B\in M_D(\C)$ is the boundary matrix, and $\Tr_D$ is the trace over the auxiliary space. 
\end{definition}
\noindent In general $B$ is allowed to depend on $\ell$ but not $x$. We say that a cMPS is bulk-uniform if $Q(x),L(x)$ are constant matrices independent of $x\in I$, and uniform if $B=\openone$. For concreteness, we use $I=[-\ell/2,+\ell/2]$. In principle one can generalize this to multiple fields \cite{haegeman2013cmps}, but for clarity we restrict to one species of bosons. 

A useful representation of the cMPS as an element of the Fock space is obtained by writing explicitly the sum over $j$-particle sectors, namely
\begin{align}
    \label{eq: cMPS-Fock}
    &\ket{\psi[B,Q,L]}\\
    &= \sum_{j=0}^\infty \int D^jx\,\Tr(BV_{-}^1L_1V_1^2...L_jV_j^{+}) \psi^\dagger(x_1)...\psi^\dagger(x_j)\ket{\Omega} \notag
\end{align}
where we use a shorthand the path-ordered measure $\int D^jx \equiv \int_I...\int_I\dd x_1...\dd x_j\,\Theta(x_1-x_2)...\Theta(x_{j-1}-x_{j})$, and the wavefunction coefficients are given in terms of $B,Q,L$: 
\begin{align}
    V_{i}^{i+1} = \mathcal{P}\exp \int_{x_i}^{x_{i+1}}\dd x\,Q(x)\,,\quad L_i \equiv L(x_i)\,.
    \label{eq: cMPS-coefficients}
\end{align}
We write $V_-^1\equiv V_{-\ell/2}^{x_1}$ and $V_j^+\equiv V_{x_j}^{+\ell/2}$ when the lower and upper limits of the integrals are at the boundaries respectively. For convenience we refer to states of the form $\int D^jx\,f(x_1,...,x_j) \psi^\dagger(x_1)...\psi^\dagger(x_j)\ket{\Omega}$ as $j$-particle Fock states.

The cMPS can be constructed as a continuum limit of discrete MPS, by setting the local dimension $d\to\infty$ and identifying the tensors at site $x$ to be $A^0_x \approx \openone + \epsilon Q(x)$, $A^{1}_x  \approx \sqrt{\epsilon} L(x)$ at leading order in lattice spacing $\epsilon$, and all $A^{n\geq 2}_x$ are fixed by physical consistency. For physically relevant situations such as computation of energy expectation values or correlation functions, in principle one may also need to impose regularity conditions for $Q(x),L(x)$ \cite{haegeman2013cmps}.

\paragraph*{Definition of cMPO.}
\label{sec: definition-obstruction}

A theory of {continuous matrix product operator} (cMPO) should provide us with a sufficiently large class of operators with the following desirable properties: 
(i) it is parametrized by some finite number of matrix-valued functions in $M_D(\C)$; 
(ii) it is obtained as a continuum limit of some MPO with constant bond dimension $D$; 
(iii) Given two cMPO $O_1,O_2$ with bond dimensions $D_1,D_2$ respectively, the product $O_1O_2$ is another cMPO with bond dimension $D\leq D_1D_2$. If $O$ is also unitary, we would also like to have 
(iv) a systematic way of verifying unitarity in terms of the local tensors. These considerations suggest the following definition of cMPO.
\begin{definition}[cMPO]
\label{definition: cMPO-ansatz}
    Let $Q(x),L(x),R(x),T(x)\in M_D(\C)$ be matrix-valued functions on $I$ and $B\in M_D(\C)$ the boundary matrix that can depend on system size $\ell \in (0,\infty]$. A cMPO with bond dimension $D$ is an operator $O$ acting on $\fock$ defined as 
    \label{eq: cMPO-ansatz}
    \begin{align}
     O \equiv O [B,Q,L,R,T] = \Tr_{D}(B\mathcal{P}e^{\int\dd x\,\mathfrak{L}_x[\cdot]})(\ketbra{\Omega}{\Omega})
    \end{align}
    where $[\cdot]$ takes operators of the field theory as input,
    \begin{equation}
        \mathfrak{L}_x \coloneqq Q(x)\otimes \Id + L(x)\otimes l_x + R(x)\otimes r_x + T(x)\otimes \Ad_x
    \end{equation}
    and the supermaps acting on the field are defined as
    \begin{equation}
    \begin{aligned}
        l_x[\cdot] &\coloneqq \psi^\dagger(x)[\cdot]\,,\quad r_x[\cdot] \coloneqq [\cdot]\psi(x)\,,\\
        \Id[\cdot] &\coloneqq [\cdot]\,,\quad\qquad\!\! \Ad_x[\cdot] \coloneqq \psi^\dagger(x)[\cdot]\psi(x)\,,
    \end{aligned}
    \end{equation}
\end{definition}
\noindent where `$\Ad_x$' refers to adjoint action by $\psi(x)$. We remark that $O$ is defined to act specifically on the Fock space $\mathfrak{F}(\mathcal{H})$ where the Fock vacuum $\ket{\Omega}$ lives for any $\ell \in (0,\infty]$, as this allows us to work with the thermodynamic limit. Once we fix $\ket{\Omega}$, all bulk-uniform cMPS in the thermodynamic limit that is not equivalent to the Fock vacuum is not in the domain of $O$. 

Let us show that Definition~\ref{definition: cMPO-ansatz} satisfies our requirements (i)-(iii). By construction (i) is satisfied. To show that it fulfills (ii), we first express a discrete MPO in a `Fock-like representation'  (See Section A of Appendix \cite{suppmat}). We first define the shorthand for the ``free propagator'' $\mathsf{W}_{x_i}^{x_j} = A^{00}_{x_i}A^{00}_{x_i+1}...A^{00}_{x_j}$ and a collection of \textit{ladder maps} at site $x$ 
\begin{align}
    f^{ij}_{x}[\cdot]\coloneqq (J^{+}_{x})^{i}[\cdot] (J^{-}_{x})^{j}\,,
    \label{eq: ladder-maps}
\end{align}
where $f^{00}_{x}[\cdot]\equiv \Id[\cdot]$ is the identity map and $J^\pm$'s are (normalized) $D$-dimensional matrix representations of the ladder operators for the $\mathfrak{su}(2)$ algebra. Then we can write any MPO as (See Section A of Appendix \cite{suppmat})
\begin{align}
    &O_N[B,\{A^{ij}_x\}] \notag\\
    &= \sum_{m=0}^N  \sum_{\substack{x_1< ... < x_m}}\!\!\!\!\!\! \Tr\rr{B\mathsf{W}_{-}^{x_1-1}A_{x_1}^{i_1j_1}...A_{x_m}^{i_m j_m} \mathsf{W}_{x_m+1}^+}\notag\\
    &\hspace{1cm}\times  f^{i_1j_1}_{x_1}...f^{i_mj_m}_{x_m}\bigr[\ketbra{\Omega_N}{\Omega_N}\bigr] + \mathcal{S}_{N,d}\,,
    \label{eq: Fock-representation-MPO}
\end{align}
where for fixed $m$ the summation contains exactly $m$ ladder operators $J^\pm_{x_k}$ and $i_k,j_k = 0,1$. The remainder $\mathcal{S}_{N,d}$ contains all terms that involves $(J^\pm_{x_k})^{n\geq 2}$ on either side of $\ketbra{\Omega_N}{\Omega_N}$. Now the strategy is to follow the cMPS construction by rescaling the matrices and operators as 
\begin{equation}
    \begin{aligned}
        A^{00}_{x} &\approx \openone_D + \epsilon Q(x) \qquad f^{00}_{x}[\cdot] =  \Id[\cdot] \,,\\
        A^{10}_{x} &\approx \sqrt{\epsilon} L(x)\!\qquad\qquad f^{10}_{x}[\cdot] \approx \sqrt{\epsilon}(\psi^\dagger(x))[\cdot]\,, \\
        A^{01}_{x} &\approx \sqrt{\epsilon} R(x)\!\qquad \qquad f^{01}_{x} [\cdot] \approx \sqrt{\epsilon}[\cdot](\psi(x))\,,\\
        A^{11}_{x}   &\approx T(x) \qquad\qquad\quad\,  f^{11}_{x} [\cdot]   \approx {\epsilon}\psi^\dagger(x)[\cdot]\psi(x)\,,
    \end{aligned}
\end{equation}
and all other matrices $A^{ij}_x$ are fixed by consistency. We then take the limit $N\to\infty,\epsilon\to 0$ while keeping $\ell$ constant, which gives ~\eqref{eq: cMPO-ansatz}. It is important that $A^{11}_x$ is $O(1)$ and not necessarily close to $\openone_{D}$ like $A^{00}_x$.

To show that $O$ satisfies (iii), we note that given two cMPOs $O_i$ with bond dimension $D_i$ ($i=1,2$) that has the form \eqref{eq: cMPO-ansatz}, the product $O=O_1O_2$ also takes the same form, with bond dimension $D\leq D_1D_2$ and the local tensors are given by 
\begin{align}
    \label{eq: product-of-MPO-tensors}
    B  &= B_1\otimes B_2\,,\qquad T = T_1\otimes T_2\,,\notag\\
    Q &= Q_1\otimes \openone_{D_2} + \openone_{D_1}\otimes Q_2 + R_1\otimes L_2\\
    L &= L_1\otimes \openone_{D_2} + T_1\otimes {L_2} \,,\quad R = \openone_{D_1}\otimes R_2 + {R_1} \otimes T_2\,.\notag
\end{align}
This can be proven from the Dyson series expansion using the product relations between $l_x,r_x,\Ad_x$ (See Section B of Appendix \cite{suppmat}). Consequently, $O_1O_2$ indeed takes the form of \eqref{eq: cMPO-ansatz} with bond dimension $D\leq D_1D_2$ and (iii) is satisfied. 

Definition~\ref{definition: cMPO-ansatz} can be used to show that a cMPO preserves the entanglement area law natively in the continuum: that is, a cMPO $O_1$ with bond dimension $D_1$ maps a cMPS $\ket{\psi[B_2,Q_2,L_2]}$ with bond dimension $D_2$ to a  cMPS $\ket{\psi[B,Q,L]}$ with bond dimension $D\leq D_1D_2$. This follows from the fact that according to Definition~\ref{definition: cMPS}, we can write the cMPS as $\ket{\psi[B_2,Q_2,L_2]}=O_2\ket{\Omega}$ for some operator $O_2$ acting on the vacuum. It can be shown that $O_2$ is a cMPO according to Definition~\ref{definition: cMPO-ansatz} (See Section C of Appendix \cite{suppmat}).

There are some gauge freedom in the definition of cMPO if we view the path-ordered exponential
\begin{align}
    W_{x}^y\coloneqq \mathcal{P}e^{\int_x^y\dd z\,\mathfrak{L}_z[\cdot]}
\end{align}
as a kind of Wilson line found in non-Abelian gauge theory, generated by matrix-valued ``gauge field'' $\mathfrak{L}_x$. A sufficient condition for two sets of matrices to generate the same cMPO is that $ \Tr[\tilde{B}\tilde{W}_-^+] = \Tr[BW_-^+]$, i.e, that the ``traced Wilson line'' are invariant under the local gauge transformation \footnote{In practice, one may also need to impose additional smoothness conditions on $g$ (such as $k$-differentiability) to ensure that the the cMPO action on a cMPS does not worsen the well-behavedness of the initial cMPS  \cite{haegeman2013cmps}.}
\begin{equation}
    \begin{aligned}
        \tilde{W}_x^y &= g(y)W_x^y g(x)^{-1}\,,\qquad g(y)\in GL_D(\C)\,,\\
        \tilde{B} &= g(x_-)Bg(x_+)^{-1}
    \end{aligned}
\end{equation}
where $x_\pm$ are the endpoints, conventionally taken to be $x_\pm = \pm \ell/2$ for a box of length $\ell$. To see how the tensors transform under the local gauge transformation, we use the fact that the operator $\mathfrak{L}_x$ is formally a generator of $W_x^y$ satisfying the first-order differential equation
\begin{align}
    \frac{\dd W}{\dd y} = \mathfrak{L}_y W\,,
\end{align}
which gives the following transformations for the tensors:
\begin{align}
    \label{eq: gauge-transformation}
    Q(x) &= g(x)^{-1}\tilde{Q}(x)g(x) + g(x)^{-1}\frac{\dd g}{\dd x}\,,\notag \\
    T(x) &= g(x)^{-1}\tilde{T}(x)g(x)\,,\\
    L(x) &= g(x)^{-1}\tilde{L}(x)g(x)\,, \quad R(x) = g(x)^{-1}\tilde{R}(x)g(x)\notag\,.
\end{align}
Notably, $Q(x)$ transforms like the vector potential in non-Abelian gauge theory.

\paragraph*{Application: cMPU.} 

As a non-trivial application of the cMPO ansatz, we show that we can construct several interesting families of cMPUs (unitary cMPOs), where unitarity is known to impose a strong constraint on the structure of MPOs. These families are natural continuum limit of non-uniform MPUs beyond the QCA family \cite{Styliaris2025matrixproduct}. In discrete settings, uniform MPUs are known to be equivalent to translationally-invariant QCAs \cite{cirac2017mpu}, and by allowing non-trivial boundary we can have MPUs beyond QCAs \cite{Styliaris2025matrixproduct}. 

The first example is the displacement operator
\begin{align}
    D(\alpha) \coloneqq e^{\int\dd x\,\alpha(x)\psi^\dagger(x) - \alpha^*(x)\psi(x)}\,,
    \label{eq: displacement-operator}
\end{align}
a $D=1$ cMPU with $B = 1$, $T = 1$, $Q = -\frac{1}{2}|\alpha(x)|^2$, $L= \alpha(x)$, and $R = -\alpha^*(x)$. This follows by noting that the supermaps $l_x,r_x,\Ad_x,\Id$ are pairwise commuting, and
\begin{align}
    \mathcal{P}e^{\int\dd x\,\Ad_x[\cdot]}[\ketbra{\Omega}{\Omega}] = \sum_{j=0}^\infty\int D^jx\,\Ad_x^j[\Omega] = \openone\,. 
    \label{eq: identity-cMPU}
\end{align}
We write $\Ad^j_x(\Omega) \coloneqq  \psi^\dagger(x_1)...\psi^\dagger(x_j)\ketbra{\Omega}{\Omega}\psi(x_j)...\psi(x_1)$ for convenience. This also shows that $\openone$ is a $D=1$ cMPU with tensors $B=T=1,Q=L=R=0$, so we know that the set of cMPUs are non-empty.

Similar to the discrete MPU setting, in general we cannot check the unitarity of a cMPO by looking at the local tensors alone. In \cite{cirac2017mpu,Styliaris2025matrixproduct}, a characterization for MPUs is given by formulating unitarity as a condition over auxiliary space. Below we provide an analogous condition.
\begin{lemma}
\label{lemma: unitarity-v1}
Let $O$ be a cMPO with bond dimension $D$ and let
\begin{align*}
    O^\dagger O &= \Tr_D(B_+\mathcal{P}e^{\int\dd x\,\mathfrak{L}_+(x)[\cdot]})[\Omega]\\
    O O^\dagger &= \Tr_D(B_-\mathcal{P}e^{\int\dd x\,\mathfrak{L}_-(x)[\cdot]})[\Omega]
\end{align*}
be two cMPOs with bond dimension $\leq D^2$ where
\begin{align*}
    \mathfrak{L}_\pm &\coloneqq Q_\pm \otimes \Id + L_\pm \otimes l_x + R_\pm \otimes r_x + T_\pm \otimes \Ad_x\,.
\end{align*}
Let $V_{x,\pm}^y$ be the free propagator generated by $Q_\pm$ and
\begin{align*}
    K^L_{j,\pm}\coloneqq L_\pm(x_j)\,,\quad K^R_{ j,\pm} \coloneqq R_\pm(x_j)\,,\quad K^A_{j,\pm}\coloneqq T_\pm(x_j)\,.
\end{align*}
Then $O$ is unitary if and only if for all $j\in \mathbb{N}\cup\{0\}$, 
\begin{align}
    \label{eq: identity-wavefunctions}
    &\Tr\rr{B_\pm V_{-,\pm}^1 K_{1,\pm}^{\alpha_1}V_{1,\pm}^2 K_{2,\pm}^{\alpha_2} \cdots K_{j,\pm}^{\alpha_j} V_{j, \pm}^+ } 
    \notag\\
    &=
    \begin{cases}
        1\qquad &\alpha_k = A\quad \forall 1\leq k\leq j\\
        0\qquad &\exists 1\leq k\leq j: \alpha_k\neq A
    \end{cases}
\end{align}
for any $x_1,x_2,...,x_j\in I$. 
\end{lemma}
\noindent This follows because if $O$ is unitary then the matrix product coefficients of both $O^\dagger O$ and $OO^\dagger$ must produce the same coefficients as the cMPU for $\openone$ in Eq.~\eqref{eq: identity-cMPU}, so it cannot contain any $L,R$ matrices and the coefficients containing only $K^{A}_\pm$ must be all equal to 1. Lemma~\ref{lemma: unitarity-v1} is simpler to use for bulk-uniform matrices. In this case, we can rewrite the matrices in the ``interaction picture'' with respect to $Q$, so that $\mathcal{K}^\alpha_\pm(x) \coloneqq V_{x_-,\pm}^{x}K_\pm^\alpha V_{x,\pm}^{x_-}$ and $\mathcal{B}_\pm(\ell)\coloneqq V_{x_-,\pm}^{x_+}B_\pm$, in which case the unitarity condition becomes
\begin{align}
    \hspace{-0.2cm}\Tr\rr{\mathcal{B}_\pm(\ell) \mathcal{K}^{\alpha_1}_\pm(x_1)...\mathcal{K}^{\alpha_j}_\pm(x_j) } 
    &=
    \begin{cases}
        1\qquad &\alpha_k = A\quad \forall k\\
        0\qquad &\exists k: \alpha_k\neq A
    \end{cases}
\end{align}

Below we consider mainly three different families of cMPUs (See Section D of Appendix \cite{suppmat} for derivations and other cMPU families). The first is the \textit{phase cMPUs} based on phase MPUs \cite{Styliaris2025matrixproduct}, which has $L=R=0$. For example, the following tensors 
\begin{align*}
    Q(x) &= i\mathrm{diag}\rr{q_1(x),...,q_D(x)}\quad T(x) \in U(1)^D \rtimes S_D\,,\\
    B &= V_{x_+}^{x_-}\cketbra{k}{+}\qquad k = 0,1,2...,D-1\,,
\end{align*}
give rise to a phase cMPU. Here $U(1)^D\equiv \bigoplus_{j=1}^D U(1)$, $S_D$ is the symmetric group of $D$ elements, $V_{x_+}^{x_-}$ is defined in \eqref{eq: cMPS-coefficients}, $q_j,t_j$ are real-valued functions. More concretely, $T(x)$ is a generalized permutation matrix $T(x) = e^{i\mathrm{diag}\rr{t_1(x),...,t_D(x)}}P$ with $P\in S_D$. Here we use rounded braket notation $\cketbra{k}{+}$ for vectors in the bond space and $\cket{+} = (1,1,1...,1)$. To get a better sense of these objects, consider the special case $D=2$ with $B\propto \cketbra{0}{+}$, $Q=i\frac{\omega}{2}Z$, $T = X$. The resulting cMPU has an explicit form involving a string operator \footnote{Here we pick the convention that $\Theta(0) = 0$ (left-continuous) so that $[\Pi(x),\psi^\dagger(x)]=0$ for all $x\in [-\ell/2,\ell/2]$.}, namely (See Section D of Appendix \cite{suppmat} for the derivation)
\begin{align}
    U_\theta = e^{-i\omega K}\,,\quad K =  \int_{x}^{\ell/2}\dd x\,x (-1)^{\Pi(x)}n(x)
\end{align}
where $\Pi(x)\coloneqq \int_x^{\ell/2}\dd z\,n(z)$ and $n(x)\coloneqq \psi^\dagger(x)\psi(x)$. 

The second family of cMPU can be obtained by modifying the phase unitary in order to go beyond diagonal unitaries. We use the fact that any qubit MPU whose vectorization admits a non-trivial compression as a locally maximally entanglable (LME) state \cite{kraus2009lme} is locally unitary-equivalent to a phase unitary \cite{Styliaris2025matrixproduct}. Since the phase cMPU $U_\theta$ is based on qubit MPUs, we can construct an infinite family of cMPUs that are not diagonal by concatenating $U_\theta$ with displacement operator that take the role of local unitaries (See Section D of Appendix \cite{suppmat}).

Last but not least, the third example we consider is based on \cite[Example 14]{Styliaris2025matrixproduct}, where a class of MPU called ``unitary action over product subspace'' is described by a $D=5$ MPU
\begin{align}
    U = \openone^{\otimes N} + \sum_{i,j=0}^1 (V_{ij}-\delta_{ij})\ket{i}^{\otimes N}\bra{j}^{\otimes N}\,,
    \label{eq: MPU-over-product-subspace}
\end{align}
where $U$ acts as $V\in U(2)$ over the two-dimensional subspace $\Span\{\ket{0}^{\otimes N},\ket{1}^{\otimes N}\}$. This was constructed as an example where bulk-uniform MPU with non-trivial boundary does not in general admit a block-diagonal decomposition for its local tensors \cite{Styliaris2025matrixproduct}. This particular example does \textit{not} generalize to the continuum, but we can modify this slightly to give a similar cMPU, namely
\begin{align}
    U = \openone + \sum_{i,j=1}^2 \rr{V_{ij}-\delta_{ij}}\ketbra{\psi_i}{\psi_j}
\end{align}
where $\ket{\psi_{i}}\equiv \ket{\psi_i[B_i,Q_i,L_i]}$ are pairwise orthonormal cMPS with bond dimension $D=D_i$. The idea is to use the cMPO ansatz \eqref{eq: cMPO-ansatz} to show that the operator
\begin{align}
    P_j = \ketbra{\psi_i[B_i,Q_i,L_i]}{\psi_j[B_j,Q_j,L_j]}
\end{align}
is a cMPO with bond dimension $D_iD_j$, so the cMPU has bond dimension $D=1+\sum_{i,j=1}^2 D_iD_j$. An analog to the above MPU in \eqref{eq: MPU-over-product-subspace} can be constructed by choosing $\ket{\psi_1} = \ket{\Omega}$ and $\ket{\psi_2}$ to be normalized one-particle Fock states $\ket{1_{f}}\coloneqq \int\dd x\,f(x)\psi^\dagger(x)\ket{\Omega}$, which is a cMPS with bond dimension $D=2$ with $B_2\propto \sigma^-,Q_2=0$ and $L_2=f(x)\sigma^+$ (See Section D of Appendix \cite{suppmat}).

\paragraph*{Discussion and outlook.} 
In this work we have introduced the family of cMPOs for bosonic quantum fields. They admit a closed-form expression in terms of path-ordered exponential of finitely many matrix-valued functions, they are obtained as a suitable continuum limit of discrete MPOs and by construction they preserve entanglement area-law directly in the continuum, and they map cMPS to cMPS. As an application we constructed a continuum limit of MPUs beyond QCA studied in \cite{Styliaris2025matrixproduct}.

Several open questions arise from these considerations. First, the cMPO ansatz \eqref{eq: cMPO-ansatz} provides a natural starting point for studying systematically the continuum limit of MPDOs \cite{verstraete2004mpdo,zwolak2004mpdo,Cuevas2013mpdo,verstraete2018convex,liu2025parent,kato2024exact}, i.e., continuous matrix product density operators (cMPDO): indeed, convex combinations of cMPS projectors form a subclass of cMPDOs. Second, since the cMPOs are in general non-Gaussian, they allow explorations into non-Gaussian operations in (1+1)-dimensional quantum field theories. Third, ideally one would like to obtain the generalization to fermionic setting natively as a continuum limit of fermionic tensor networks \cite{piroli2021fermionic,verstraete2024fermionic}, and also generalization to higher dimensions so that it preserves area-law entanglement of continuous tensor networks \cite{tilloy2019cTNS}. Last but not least, it would be interesting to see if it is possible to formulate a notion of cMPO symmetry using the cMPO ansatz.

We mention in passing that the cMPOs defined in this work seem to exclude some physically interesting operators in the context of quantum field theory, such as translation or unitaries generated by field derivatives. In principle, the ansatz can be generalized to cover such operators, although it is no longer automatic that the resulting cMPO arises from some continuum limit of discrete MPO. We provide some heuristic examples in Section E of Appendix \cite{suppmat} and we leave the general theory of such ``generalized cMPOs'' for future work.

\paragraph*{Acknowledgment.} The authors are grateful to Rahul Trivedi, Georgios Styliaris, Marta Florido-Llinás for insightful discussions. E.T. acknowledges support from the Alexander von Humboldt Foundation through the Humboldt Research Fellowship. J.I.C acknowledges funding by THEQUCO as part of the Munich Quantum Valley, which is supported by the Bavarian state government with funds from the Hightech Agenda Bayern Plus.

\bibliography{ref}

\appendix

\section{Fock representation of MPO and cMPO ansatz}

In this section we first review how the Fock representation works for MPS in order to see how it works for MPO. Assume that we have an $N$-qudit system on a lattice with spacing $\epsilon$ so that $\ell=N\epsilon$ is the total length or volume of the system. Let $\ket{\Omega_N} \coloneqq \ket{0}^{\otimes N}$ be the all-zero (``vacuum'') state and let $J^\pm$ be (normalized) ladder operators of $\mathfrak{su}(2)$ algebra where in the spin $j=\frac{1}{2}(D-1)$ representation we have $J^\pm \in M_D(\C)$. 

Let us use the shorthand $\mathsf{V}_{i}^{j}\coloneqq A^{0}_{x_i}A^{0}_{x_i+1}...A^{0}_{x_j}$ which takes the role of the free propagator $V_{x}^y$ for cMPS. We will use the notation $A^{i}_j\equiv A^{i}_{x_j}$ interchangeably in the discrete case. We can then express a qudit cMPS as sums over $N$-particle sectors
\begin{align}
    \ket{\psi[B,\{A_x^{i}\}]} &\propto \sum_{i_1...i_N}\Tr(BA^{i_1}_1...A^{i_N}_N)(J^{+}_1)^{i_1}...(J^{+}_{N})^{i_N}\ket{\Omega_N}\notag\\
    &\equiv \sum_{n=0}^N \ket{\Psi_n}
\end{align}
where $\ket{\Psi_n}\in \mathcal{H}^{\odot n}$, i.e., they are vectors in the symmetric $n$-particle subspace of the Fock space $\mathfrak{F}(\mathcal{H})$. The first three sectors are
\begin{equation}
\label{eq: MPS-n-particle-sectors}
\begin{aligned}
    \ket{\Psi_0} &= \Tr(B\mathsf{V}_1^N)\ket{\Omega_N}\\
    \ket{\Psi_1} &= \sum_{j=1}^N \Tr({B\mathsf{V}_{1}^{j-1}A^{1}_j \mathsf{V}_{j+1}^N}) J^+_j\ket{\Omega_N}\\
    \ket{\Psi_2} &= \sum_{i < j= 1}^N \Tr({B\mathsf{V}_{1}^{i-1}A^{1}_i \mathsf{V}_{i+1}^{j-1}A^{1}_j  \mathsf{V}_{j+1}^N }) J^+_iJ^+_j\ket{\Omega_N}\\
    &\quad + \sum_{j=1}^N \Tr({B\mathsf{V}_{1}^{j-1}A^{2}_j \mathsf{V}_{j+1}^N}) (J^+_j)^2\ket{\Omega_N}\,,
\end{aligned}
\end{equation}
and the higher particle sectors proceed similarly. Essentially, each vector $\ket{\Psi_n}$ contains every term in the MPS ansatz that contains $n$ creation operators applied to the vacuum $\ket{\Omega_N}$. 

As stated, in Eq.~\eqref{eq: MPS-n-particle-sectors} we have only partitioned the standard MPS ansatz into $N+1$ partitions based on how the particle content. The next step is to rewrite Eq.~\eqref{eq: MPS-n-particle-sectors} in a form that closely parallels the cMPS ansatz \eqref{eq: cMPS-Fock} as path-ordered sums. For $n=0,1$ they are already in the correct form, so let us do this for $n=2$ and $n\geq 3$ works similarly. Then we have 
\begin{align}
    \ket{\Psi_2} &= \sum_{j=1}^N \sum_{i<j} \Tr({B\mathsf{V}_{1}^{i-1}A^{1}_i \mathsf{V}_{i+1}^{j-1}A^{1}_j  \mathsf{V}_{j+1}^N }) J^+_iJ^+_j\ket{\Omega_N}\notag\\
    &\quad + \sum_{j=1}^N  \Tr({B\mathsf{V}_{1}^{j-1}A^{2}_j \mathsf{V}_{j+1}^N}) (J^+_j)^2\ket{\Omega_N} 
\end{align}
The full qudit MPS can thus be written as 
\begin{align}
   \ket{\psi[B,\{A_x^{i}\}]}
   &= \sum_{m=0}^{N} \sum_{i_1< ... < i_m}\!\!\!\!\!\!\! { c_{i_1...i_m}} J^+_{1}\cdots J_{m}^+\ket{\Omega_N}\notag \\
   & \quad + \mathcal{R}_{N,d}\,,
   \label{eq: Fock-representation-MPS}
\end{align}
where
\begin{align}
    c_{i_1...i_n} \coloneqq \Tr(BV_{1}^{i_1-1}A^1_{i_1}V_{i_1+1}^{i_2-1}A^1_{i_2}...A^{1}_{i_n}V_{i_n+1}^{N})
\end{align}
and $\mathcal{R}_{N,d}$ contains all remaining terms that depend on at least one $A^{i\geq 2}_{x_j}$, such as
\begin{align*}
     &\Tr\rr{B\mathsf{V}_{1}^{i-1}A^{2}_{i} \mathsf{V}_{i+1}^N} (J^+_{i})^2\ket{\Omega_N}\\
     &\Tr\rr{B\mathsf{V}_{1}^{j-1}A^{3}_{j} \mathsf{V}_{j+1}^{k-1}A^{1}_{k} \mathsf{V}_{k+1}^{N}} (J^+_{j})^3J^+_{k}\ket{\Omega_N}\,.
\end{align*}
We refer to Eq.~\eqref{eq: Fock-representation-MPS} as the \textit{Fock representation} of the MPS $\ket{\psi[B,{A^i_x}]}$. 

The continuum limit for MPS is defined by rescaling
\begin{equation}
    \begin{aligned}
        A^0_{x} &\sim \openone + \epsilon Q(x)\,,\\
        A^{1}_{x}&\sim \sqrt{\epsilon}L(x)\,,\qquad J^+_{x}  \sim \sqrt{\epsilon}\psi^\dagger(x)\,,
    \end{aligned}
\end{equation}
and then take the limit $N\to\infty,\epsilon\to 0$ while keeping $\ell$ constant. The continuum limit of the all-zero state $\ket{\Omega_N}\coloneqq \ket{0}^{\otimes N}$ is given by the Fock vacuum $\ket{\Omega_N}\to\ket{\Omega} \in \mathfrak{F}(\mathcal{H})$ which satisfies $\psi(x)\ket{\Omega_N} = 0$ for all $x$. When we take the continuum limit, consistency requires that in every $m$-particle sector we identify
\begin{align}
    \lim_{\epsilon\to 0}&\sum_{i_1< ... < i_m}{ c_{i_1...i_n}} J^+_{i_1}\cdots J_{i_m}^+\ket{\Omega_{N}} \notag\\
    &\sim \frac{1}{m!}\sum_{i=1}^N \Tr(B\mathsf{V}_{1}^{i}(A^1_i)^m\mathsf{V}_{i}^{N})(J^+_{i})^m\ket{\Omega_N} \notag\\
    &\equiv \sum_{i=1}^N \Tr(B\mathsf{V}_{1}^{i}A^m_j\mathsf{V}_{i}^{N})(J^+_{i})^m\ket{\Omega_N} 
\end{align}
which enforces that the local dimension becomes infinite-dimensional and that
\begin{align}
    A^{m}_j = \frac{1}{m!}(A^1)^m \sim \frac{1}{m!}(\sqrt{\epsilon}L(x_j))^m\,.
\end{align}
Consequently, the only freedom we have for the continuum limit of a cMPS is in specifying $A^0,A^1$. This can be interpreted as saying that a cMPS is naturally a continuum limit of qubit MPS as the bond matrices $A^{n\geq 2}$ are not specified independently. This does not imply that qudit MPS has no continuum limit: it means that Definition~\ref{definition: cMPS} needs to be generalized, see \cite{gemma2018continuum,gemma2020generalizedMPS}.

For MPOs it is slightly more involved. Using the ladder maps Eq.~\eqref{eq: ladder-maps}, we can write any MPO as 
\begin{align}
    &O_N[B,\{A^{ij}_x\}] \notag\\
    &= \sum_{m=0}^N  \sum_{\substack{x_1< ... < x_m}}\!\!\!\!\!\! \Tr\rr{B\mathsf{W}_{1}^{x_1-1}A_{x_1}^{i_1j_1}...A_{x_m}^{i_m j_m} \mathsf{W}_{x_m+1}^m}\notag\\
    &\hspace{3cm}\times  f^{i_1j_1}_{x_1}...f^{i_mj_m}_{x_m}\bigr[\ketbra{\Omega_N}{\Omega_N}\bigr] \notag\\
    &\quad + \mathcal{S}_{N,d}\,,
\end{align}
which is Eq.~\eqref{eq: Fock-representation-MPO}, where for fixed $m$ the summation contains exactly $m$ ladder operators $J^\pm_{x_k}$ and $i_k,j_k = 0,1$. The remainder $\mathcal{S}_{N,d}$ contains all terms that depend on at least one $A^{ij}_{x_j}$ with either $i,j\geq 2$, such as
\begin{align*}
    &\Tr\rr{B\mathsf{W}_{1}^{x_j-1}A^{20}_{x_j} \mathsf{W}_{x_j+1}^{x_k-1}A^{01}_{x_k}\mathsf{W}_{x_k+1}^{N}} (J^+_{x_j})^2\ketbra{\Omega_N}{\Omega_N}(J^-_{x_k}) \\
    &\Tr\rr{B\mathsf{W}_{1}^{x_j-1}A^{13}_{x_j} \mathsf{W}_{x_j+1}^N} (J^+_{x_j})\ketbra{\Omega_N}{\Omega_N}(J^-_{x_j})^3\\
    &\Tr\rr{B\mathsf{W}_{1}^{x_j-1}A^{22}_{x_j} \mathsf{W}_{x_j+1}^N} (J^+_{x_j})^2\ketbra{\Omega_N}{\Omega_N}(J^-_{x_j})^2\,.
\end{align*}
We refer to Eq.~\eqref{eq: Fock-representation-MPO} also as the Fock representation of the MPO $O[B,{A^{ij}_x}]$.

Following the same argument for the cMPS, it follows that the continuum limit for MPO is defined by rescaling 
\begin{equation}
    \begin{aligned}
        A^{00}_{x} &\approx \openone_D + \epsilon Q(x) \qquad f^{00}_{x}[\cdot] =  \Id[\cdot] \,,\\
        A^{10}_{x} &\approx \sqrt{\epsilon} L(x)\!\qquad\qquad f^{10}_{x}[\cdot] \approx \sqrt{\epsilon}(\psi^\dagger(x))[\cdot]\,, \\
        A^{01}_{x} &\approx \sqrt{\epsilon} R(x)\!\qquad \qquad f^{01}_{x} [\cdot] \approx \sqrt{\epsilon}[\cdot](\psi(x))\,,\\
        A^{11}_{x}   &\approx T(x) \qquad\qquad\quad\,  f^{11}_{x} [\cdot]   \approx {\epsilon}\psi^\dagger(x)[\cdot]\psi(x)\,,
    \end{aligned}
\end{equation}
and for $i+j \geq  3$ we have
\begin{equation}
\begin{aligned}
    f^{ij}_{x} [\cdot]   &\sim \sqrt{\epsilon}^{i+j}(\psi^\dagger(x))^{i}[\cdot](\psi(x))^{j}\,.
\end{aligned}
\end{equation}
We then take the limit $N\to\infty,\epsilon\to 0$ while keeping $\ell$ constant. Crucially, $A^{11}_x$ scaling is required to be $\mathcal{O}(1)$ and is not necessarily close to $\openone_D$ unlike $A^{00}_x$. Just like the cMPS, the scalings for $A^{ij}_{x}$ with $i+j\geq 3$ are fixed consistency in the continuum limit and only $A^{00}_x,A^{10}_x,A^{01}_x,A^{11}_x$ are independent matrices, so this prescription treats $A^{mn}_x$ for $m+n\geq 3$ as being obtained from only appropriate combinations of these four sets of matrix-valued functions. Under such prescription, the continuum limit is given by Eq.~\eqref{eq: cMPO-ansatz}. Since the only freedom we have for to take the continuum limit is in specifying four bond matrices $A^{00},A^{01},A^{10},A^{11}$ of the corresponding MPO, this can be interpreted as saying that a cMPO is naturally a continuum limit of {qubit} MPO. In principle, this also does not imply that only qubit MPOs have a continuum limit. We leave the possibility of generalizing the cMPO ansatz to allow for continuum limit of qudit MPOs in the same way it was done for cMPS in \cite{gemma2020generalizedMPS} to future work.

\section{Product relations} 
To show that the product of two cMPO is another cMPO with bond matrices given by \eqref{eq: product-of-MPO-tensors}, we rely on the fact that only these products contribute in the Dyson series expansion:
\begin{align}
    \label{eq: field-projectors}
    \Omega\cdot r_x(\Omega) &= r_x(\Omega)\,, \quad \Omega\cdot\Omega = \Omega\,,\notag\\
    l_x(\Omega)\cdot r_y(\Omega) &= l_x r_y(\Omega)\,,\quad 
    l_x(\Omega)\cdot\Omega = l_x(\Omega)\notag\\
    r_x(\Omega)\cdot l_y(\Omega) &= \delta_{xy}\Omega\,,\quad
    r_x(\Omega)\cdot\Ad_y(\Omega) = \delta_{xy}r_x(\Omega)\,,\notag\\
    \Ad_x(\Omega)\cdot l_y(\Omega) &= \delta_{xy}l_x(\Omega)\,,\\ 
    \Ad_x(\Omega)\cdot\Ad_y(\Omega) &= \delta_{xy}\Ad_x(\Omega)\notag\,.
\end{align}
We used the shorthand $\Omega\coloneqq \ketbra{\Omega}{\Omega}$ and $\delta_{xy}\equiv \delta(x-y)$. 

\section{Area-law preservation for cMPS} To show that a cMPO preserves entanglement area-law directly in the continuum, we first show that an operator $O'$ of the form
\begin{align*}
    {O}' = \Tr_D(B\mathcal{P}e^{\int_I \dd x\,Q(x)\otimes \openone + L(x)\otimes \psi^\dagger(x) + R(x) \otimes \psi(x)})
\end{align*}
can be recast into the form given in Eq.~\eqref{eq: cMPO-ansatz}, by choosing $B,Q,L,R\in M_D(\C)$ to be the same as those that appear in $O'$ and setting $T = \openone_D$. To see this, note that since the ladder maps commute, i.e.,
\begin{align}
    [l_x,\Ad_x] = [r_x,\Ad_x] = [l_x,r_x] = 0\,, 
    \label{eq: ladder-maps-commute}
\end{align}
the path-ordered exponential in the ansatz \eqref{eq: cMPO-ansatz} can be decomposed into two parts, namely 
\begin{align*}
    \mathcal{P}e^{\int\dd x\,\mathfrak{L}_x[\cdot]} &= \mathcal{F} \circ \mathcal{F}_{\Ad}\,,\\
    \mathcal{F}_{\Ad} [\cdot] &\coloneqq \mathcal{P}e^{\int\dd x\,\openone_D\otimes \Ad_x}[\cdot] = \openone_D\otimes\openone\,,\\
    \mathcal{F}[\cdot] &\coloneqq  \mathcal{P}e^{\int\dd x\,Q(x)\otimes \Id[\cdot] + L(x)\otimes l_x[\cdot]+R(x)\otimes r_x[\cdot]}\,.
\end{align*}
Using Dyson series expansion we see that $O'=\Tr_D(B\mathcal{F} [\openone_D\otimes\openone])$ and it allows us to write a cMPS as
\begin{align}
    \ket{\psi[B_2,Q_2,L_2]} \equiv O_{2}[B_2,Q_2,{L_2,0,\openone_{D_2}}]\ket{\Omega}\,,
\end{align}
which shows that a cMPS is a cMPO acting on $\ket{\Omega}$. Using Eq.~\eqref{eq: field-projectors} we can compute the cMPO for $O_1O_2$ with local tensors given according to the product rule \eqref{eq: product-of-MPO-tensors}. Observe that the tensors $R=R_1\otimes\openone_{D_2}$ and $T = T_1\otimes \openone_{D_2}$ associated with $O = O_1O_2$ do not contribute when acting on $\ket{\Omega}$. It follows that the resulting cMPS $\ket{\psi[B,Q,L]} = O_1O_2\ket{\Omega}$ has local tensors
\begin{equation}
\begin{aligned}
    B &= B_1\otimes B_2\,,\quad L = L_1\otimes \openone_{D_2} + {T_1 \otimes L_2}\\
    Q &= Q_1\otimes \openone_{D_2} + \openone_{D_1}\otimes Q_2 + R_1\otimes L_2\,.
\end{aligned}
\end{equation}
The bond dimension of the resulting cMPS is $D\leq D_1D_2$.

\section{Construction of cMPU families }

\subsection{Phase MPU and cMPU}
\label{sec: cMPU}

Here we will consider several families of cMPUs including the examples given in the main text and provide further details of their constructions. The first natural family of cMPUs is based on a subclass of diagonal unitaries called \textit{phase unitaries} \cite{Styliaris2025matrixproduct}. These will help us construct a large family of non-trivial cMPUs with $D\geq 2$. 
    
\subsubsection*{Phase MPU}

Our starting point is the concept of locally maximally entangleable (LME) states first introduced in \cite{kraus2009lme}.
\begin{definition}[LME state \cite{kraus2009lme}]
    A multipartite state $\ket{\Psi_{\mathsf{LME}}}$ is LME if there exists isometries 
    \begin{align*}
        V_j: \mathcal{H}'_{j} \to \mathcal{H}_{A,j}\otimes \mathcal{H}_{B,j}\,,
    \end{align*}
    where $\dim \mathcal{H}'_j=d_j'$ and $\dim \mathcal{H}_{\alpha, j} = d_j$ for $\alpha=A,B$, with $V_j^\dagger V_j = \openone_{d_j'}$ such that the state
    \begin{align}
        \ket{\tilde{\Psi}_\mathsf{LME}}_{AB} = \bigotimes_j V_j  \ket{\Psi_{\mathsf{LME}}}
    \end{align}
    is maximally entangled across the $AB$ bipartition. If $d_j=d,d_j'=d'$ for all $j$ then we say that $\ket{\Psi_{\mathsf{LME}}}$ is $(d,d')$-LME.
\end{definition}
\noindent Essentially, LME states are those that can be transformed to maximally entangled state by first appending local auxiliary degrees of freedom and then perform local unitaries on each system-ancilla pair. All product states have this property, but in general this is a non-trivial constraint on multipartite systems. 

The relevant object for us is the family of $(2,2)$-LME states that has been completely characterized in \cite{kraus2009lme} if we assume that $V_j$'s are control isometries (or equivalently, control unitaries with one input state fixed). This family was further generalized to states where $V_j$'s can be chosen to be any entangling isometries in \cite{Styliaris2025matrixproduct}. 
\begin{proposition}[Phase-LME \cite{kraus2009lme,Styliaris2025matrixproduct}]
    \label{proposition: (2,2)-lme}
    Up to local-unitary (LU) transformations, all $(2,2)$-LME states for which $V_j$'s can be chosen to be entangling isometries have the form
    \begin{align}
        \ket{\Psi_{\theta}} = \frac{1}{\sqrt{2^N}}\sum_{i_1=0}^1...\sum_{i_N=0}^1e^{i\theta_{i_1...i_N}}\ket{i_1...i_N} 
    \end{align}
    and conversely, every such state is $(2,2)$-LME where $V_j$'s can be chosen to be entangling isometries. We call these states phase-LME states.
\end{proposition}
For an $N$-qubit system, we can always vectorize any unitary $U$ to obtain
\begin{align}
    \ket{U_{AB}} = \sum_{i,j} U_{ij}\ket{i}_A\ket{j}_B\,.
\end{align}
This is an unnormalized maximally entangled state, i.e., its Choi-Jamio\l kowski  state. We can always write this in terms of the local isometric compression
\begin{align}
    \ket{U_{AB}} = \bigotimes_j V_j\ket{\Psi_U}\,,
\end{align}
where $V_j$ are local isometries and the state $\ket{\Psi_U}$ is (unnormalized) LME state. Clearly, every unitary $U$ admits such a representation in the trivial sense $V_j = \openone_{d^2}$ (when $V_j$'s are local unitaries). A more interesting case is when a non-trivial compression is possible, i.e., when $\ket{\Phi_U}$ is strictly lower-dimensional than $\ket{U_{AB}}$ we started with. A subclass of unitaries that admits a genuine non-trivial isometric compression is the {phase unitaries}.

\begin{definition}[Phase unitaries \cite{Styliaris2025matrixproduct}]
\label{eq: phase-unitary-general}
    A phase unitary is a diagonal unitary given by
    \begin{align}
        U_{\theta} = \sum_{i_1=1}^d...\sum_{i_N=1}^de^{i\theta_{i_1...i_N}}\ketbra{i_1...i_N}{i_1...i_N}\,.
    \end{align}
    That is, its Choi-Jamio\l kowski state is obtained from the unnormalized phase-LME state $\ket{\Psi_\theta}$
    \begin{align}
        \ket{U_\theta} &= \bigotimes_j V_j\ket{\Psi_\theta}\,,\notag\\
        \ket{\Psi_\theta} &=\sum_{i_1=1}^d...\sum_{i_N=1}^de^{i\theta_{i_1...i_N}}\ket{i_1...i_N} \label{eq: phase-LME}
    \end{align}
    where $V_j:\ket{i_j}\mapsto \ket{i_ji_j}$. 
\end{definition}

Up to this point, we do not have any tensor-network assumptions. We are interested 
in a subclass of phase unitaries that are also MPUs with bond dimension $D$, which was first studied in \cite{Styliaris2025matrixproduct}. 

\begin{definition}[Phase MPS \cite{Styliaris2025matrixproduct}]
    \label{def: phase-MPU}
    A {{phase MPS}} with bond dimension $D$ is a phase-LME state $\ket{\Psi_\theta}$ that is also an MPS with bond dimension $D$: that is,
    \begin{subequations}
    \label{eq: phase-MPS}
    \begin{align}
        \ket{\Psi_{\theta}} = \sum_{i}\Tr(BA^{i_1}_1...A_N^{i_N})\ket{i_1...i_N}\,,
    \end{align}
    where $B, A^{i}_k\in \M_{D}(\C)$ for all $k$ and $i$ such that
    \begin{align}
        \quad \Tr(BA_1^{i_1}...A_N^{i_N})= e^{i\theta_{i_1...i_N}}\,.
    \end{align}
    \end{subequations}
\end{definition}
\noindent Phase MPS provides us immediately with a family of phase MPUs using local copy isometry $V_j:\ket{i_j}\mapsto \ket{i_ji_j}$.
\begin{definition}[Phase MPU \cite{Styliaris2025matrixproduct}]
A \emph{{phase MPU}} is a diagonal phase unitary such that
\begin{align}
    U_\theta = \sum_{i}\Tr(BA^{i_1i_1}_1...A_N^{i_Ni_N})\ketbra{i_1...i_N}{i_1...i_N}\,.
\end{align}
with $\Tr(BA_1^{i_1i_1}...A_N^{i_Ni_N})= e^{i\theta_{i_1...i_N}}$. That is, $U_\theta$ is the unitary whose Choi-Jamio\l kowski state is $\ket{U_\theta} = \bigotimes_j V_j\ket{\Psi_\theta}$, where $V_j:\ket{i_j}\mapsto \ket{i_ji_j}$ is local isometry and $\ket{\Psi_\theta}$ is the phase MPS \eqref{eq: phase-MPS}, by identifying $A^{ii}_k$ with $A^{i}_k$ from the phase MPS.   
\end{definition}

\begin{example}[Weighted finite automata \cite{Styliaris2025matrixproduct}]
    \label{example: multi-control Z-gate}
    Consider {a phase MPS} \eqref{eq: phase-MPS} with non-uniform bulk tensors 
    \begin{align}
        \cbraket{\alpha|A^{i}_k|\beta} = \delta_{\beta,f^{(k)}_{\alpha,i}}e^{i\theta_{\alpha,i}^{(k)}}\,.
    \end{align}
    where $f^{(k)}_{\alpha,i}\in \{0,1,2,...,D_{k-1}-1\}$ and $\theta_{\alpha,i}^{(k)}$ are arbitrary phases. This is essentially a deterministic weighted finite automaton with complex phases as weights \cite{Styliaris2025matrixproduct}, memory size fixed by $D$ and $f$ specifies how the transition to different automaton states works. This gives us a phase MPU according to Definition~\eqref{def: phase-MPU}. A concrete example is the multi-control $Z$-gate
    \begin{align}
        U = \openone^{\otimes N} - 2\ketbra{1}{1}^{\otimes N}\,,
    \end{align}
    where $D_k=2$ for all $k$, $f_{\alpha,i} = i\cdot\alpha$ and $\theta^{1}_{1,1} = \pi$ is the only non-trivial phase.
    \exampleqed
\end{example}

\subsubsection*{Phase cMPU}

We are now ready to construct our first non-trivial family of cMPUs by following the construction of phase MPUs. The first step is to take unnormalized phase MPS state in Eq.~\eqref{eq: phase-LME} and consider its continuum limit, i.e., a \textit{phase cMPS}. 
\begin{definition}[Phase cMPS]
    We say that $\ket{\Phi_\theta}$ is a {{phase cMPS}} if it is a cMPS $\ket{\psi[B,Q,L]}$ defined in Eq.~\eqref{eq: cMPS-Fock} such that the continuous matrix-product coefficients are complex phases. That is, for all $j\geq 0$ we have
    \begin{align}
        \Tr(BV_-^1L_1V_1^2...V_{j-1}^jL_jV_j^+ ) = e^{i\theta_j(x_1,...,x_j)}
        \label{eq: complex-phase-MPS}
    \end{align}
    for some choice of matrices $B,Q(x),L(x)\in M_D(\C)$ and a family of phase functions $\{\theta_j: [0,\ell]^j\to \R \}$. In general we allow $B$ to depend on $\ell$. 
\end{definition}
\noindent For now let us assume that we can find non-trivial $B,Q,L\in M_D(\C)$ such that Eq.~\eqref{eq: complex-phase-MPS} holds. Then starting from the phase cMPS expression
\begin{align}
    \ket{\Phi_\theta}
    &=  \sum_{j=0}^\infty \int D^jx\,e^{i\theta_j(x_1...x_j)}\psi^\dagger(x_1)...\psi^\dagger(x_N)\ket{\Omega}\,,
    \label{eq: phase-cMPS}
\end{align}
we formally adapt the local copy isometry $V = \bigotimes_j V_j$ to the continuum limit by considering the mapping
\begin{align}
    V:
    \prod_{i=1}^j \psi^\dagger(x_i)\ket{\Omega}\mapsto \prod_{i=1}^j \psi^\dagger(x_i) \ketbra{\Omega}{\Omega}\psi(x_i)
    \label{eq: local-isometry-continuum}
\end{align}
and apply to $V$ to $\ket{\Psi_\theta}$ in Eq.~\eqref{eq: phase-cMPS}. This provides us with a family of {phase cMPU} $U_\theta $ associated with the phase cMPS $\ket{\Phi_\theta}$ in analogy with the discrete phase MPU.

\begin{lemma}[Phase cMPU]
\label{lemma: phase cMPU}
The following diagonal phase unitary 
\begin{align}
    U_\theta
    &\coloneqq  \sum_{j=0}^\infty \int D^jx\,\Tr(BV_-^1T_1V_1^2...V_{j-1}^jT_jV_j^+ )\notag\\
    &\qquad \psi^\dagger(x_1)...\psi^\dagger(x_j)\ketbra
    {\Omega}{\Omega}\psi(x_j)...\psi(x_1)\,,
    \label{eq: phase cMPU-Fock}
\end{align}
obtained from un-vectorizing $\ket{U_\theta}$ such that the continuous matrix-product coefficients \eqref{eq: complex-phase-MPS} holds, is a cMPU of the form
\begin{align}
    U_\theta = \Tr_D\rr{B\mathcal{P}e^{\int \dd x \,Q(x)\otimes \Id[\cdot] + T(x)\otimes \Ad_x[\cdot] }}\rr{\ketbra{\Omega}{\Omega}}
    \label{eq: phase cMPU-exp}
\end{align}
where $Q,T$ satisfy Eq.~\eqref{eq: complex-phase-MPS} for all $j \geq 0$. The matrix $T(x)$ is identified with $L(x)$ of the corresponding phase cMPS.
\end{lemma}

\begin{proof}
    The proof follows directly from the fact that diagonal unitaries can only contain the $\Ad_x$ supermaps. 
\end{proof}

The existence of phase cMPUs in Lemma~\ref{lemma: phase cMPU} relies on the existence of phase cMPS, which in turn depends on whether it is possible to find matrices $B,Q,L\in M_D(\C)$ satisfying Eq.~\eqref{eq: complex-phase-MPS}. Here we give an explicit construction of a family of phase cMPU for every $D\geq 1$ {mentioned in the main text}. 
\begin{proposition}[permutation-phase cMPU]
    \label{proposition: generalized-permutation-cMPU}
    The following family of matrices
    \begin{align*}
        Q(x) &= i\mathrm{diag}\rr{q_1(x),...,q_D(x)}\,,\\
        T(x) &\in U(1)^D \rtimes S_D\,,\\
        B &= V_{x_+}^{x_-}\cketbra{k}{+}\qquad k = 0,1,2...,D-1\,,
    \end{align*}
    give rise to a phase cMPU. Here $U(1)^D\equiv \bigoplus_{j=1}^D U(1)$ is a direct sum of phases, $S_D$ is the symmetric group of $D$ elements, $V_{x_+}^{x_-}$ is defined in Eq.~\eqref{eq: cMPS-coefficients}, $q_j,t_j$ are real-valued functions, and $T(x)$ is a generalized permutation matrix $T(x) = e^{i\mathrm{diag}\rr{t_1(x),...,t_D(x)}}P$ with $P\in S_D$. Here we use rounded braket notation $\cketbra{k}{+}$ for vectors in the bond space and $\cket{+} = (1,1,1...,1)$. 
\end{proposition}

\begin{proof}
    Since $V_x^y = e^{i \int_x^y\dd z\, Q(z)} \in U(1)^D \rtimes S_D $, it follows that product 
    \begin{align}
        V_-^1T_1V_1^2...V_{j-1}^jT_jV_j^+ \in U(1)^D\rtimes S_D\,,
    \end{align}
    i.e.,  the matrix product is closed for arbitrary $j\geq 0$. The crucial part is the choice of boundary: in order to guarantee that we extract the resulting phase for all $j$ without fail, we need $B\propto \cketbra{k}{+}$ for $k\in \{0,1,2,...,D-1\}$ to ensure that exactly one phase factor is picked up by the trace over the auxiliary space, thus ensuring unitarity. It can also be checked that this choice indeed fulfills unitarity condition in Lemma~\ref{lemma: unitarity-v1}.
\end{proof}

To have a more concrete description of these unitaries, we first look at two concrete examples for $D=1$ and $D=2$.
\begin{example}
    \label{example: phase-D=1}
    Set $D=1$. Then $B,Q,T$ are scalars and
    \begin{align}
        U_\theta &= \sum_{j=0}^\infty \int D^jx\, B e^{\int_I\dd x\, Q(x)}T(x_1)...T(x_j)\Ad^j_x(\Omega)
    \end{align}
    which is a phase unitary if and only if
    \begin{align}
        Q(x) = i q(x)\,,\qquad T(x) = e^{ir(x)}
    \end{align}
    where $q,r$ are real-valued. Since $q(x)$ only introduces global phase, we can set $B = e^{-i\int_I \dd x\, q(x)}$ to absorb the global phase and since everything commutes, we can write
    \begin{align}
        U_\theta &= \sum_{j=0}^\infty \int \frac{\dd^j x}{j!}\, e^{i\sum_{i=1}^j r(x_j)}\Ad^j_x(\Omega)\,.
        \label{eq: D=1 phase-cMPU}
    \end{align}
    We can see that this is a phase unitary by looking at its action on each $N$-particle sector.
    
    While Eq.~\eqref{eq: D=1 phase-cMPU} is sufficiently concrete, in quantum field theory it is desirable to have a more ``native'' expression in terms of the field operators $\psi(x)$. This is indeed possible, and we get
    \begin{align}
        U_\theta &= \exp \rr{i\int_I\dd x \, r(x)n(x)}
    \end{align}
    where $n(x)=\psi^\dagger(x)\psi(x)$ is the number density operator, which is manifestly unitary. As an aside, the identity operator $U = \openone$ can thus be seen as a trivial phase cMPU by setting $Q(x) = 0$ and $T(x) = 1$.
    \exampleqed 
\end{example}

To facilitate the subsequent constructions, we first prove the following lemma.
\begin{lemma}
    \label{lemma: left-continuity-Heaviside}
    Let $n(x)=\psi^\dagger(x)\psi(x)$ and define 
    \begin{align}
        \Pi(x) \coloneqq \int_x^{\ell/2}\dd z\,n(z)\,.
    \end{align}
    Then
    \begin{align}
         e^{i\theta \Pi(x)}\psi^\dagger(x)e^{-i\theta\Pi(x)} &= e^{i\theta \Theta(0)}\psi^\dagger(x)
    \end{align}
    where $\Theta(0)$ is an extended value of the Heaviside function $\Theta$ at the origin. In particular,  if we extend $\Theta$ to be left-continuous, i.e., $\Theta(0)=0$, then $[\Pi(x),\psi^\dagger(x)]=0$.
\end{lemma}
\begin{proof}
    We have
    \begin{align}
         [\Pi(x),\psi^\dagger(y)] 
         &= \int_x^{\ell/2}\dd z\,[n(z),\psi^\dagger(y)] \notag\\
         &= \int_x^{\ell/2}\dd z\,\delta(y-z)\psi^\dagger(y)\notag\\
         &= \Theta(\ell/2-y)\Theta(y-x)\psi^\dagger(y)\notag\\
         &= \Theta(y-x)\psi^\dagger(y)\,.
    \end{align}
    where in the last equality we use the fact that $y\leq \ell/2$. The standard Baker-Campbell-Hausdorff (BCH) formula then gives in the coincidence limit $y\to x$
    \begin{align}
        e^{i\theta \Pi(x)}\psi^\dagger(x)e^{-i\theta\Pi(x)} &= e^{i\theta \Theta(0)}\psi^\dagger(x)\,.
    \end{align}
    We pick the convention that $\Theta(0) = 0$, i.e., that $\Theta$ is left-continuous, so that $[\Pi(x),\psi^\dagger(x)]=0$ for all $x\in [-\ell/2,\ell/2]$.
\end{proof}

\begin{example}[Parity-controlled phase unitary]
    \label{example: phase-D=2}
    Consider a phase unitary $U_\theta$ with $D=2$, specified by bulk-uniform matrices $Q,T$ and $\ell$-dependent boundary matrix $B$ 
    \begin{align}
        B = V_{x_+}^{x_-}\cketbra{0}{+}\,,\quad  Q = \frac{\omega}{2}Z\,,\quad T = X\,,
    \end{align}
    where $X,Z$ are Pauli matrices and $\omega >0$. By direct computation, one can check that the phase functions for any $N\geq 0$ reads
    \begin{align}
        \theta_N(x_1,...,x_N) = \begin{cases}
            1 &\qquad N = 0\\
            -\sum_{j=1}^N (-1)^{N-j} x_j &\qquad N\geq 1
        \end{cases}
        \label{eq: D=2-phases}
    \end{align}
    The dependence of $B$ on system size $\ell$ is necessary to make the phases independent of $\ell$, which is relevant if we are interested in the thermodynamic limit $\ell\to\infty$. The choice is not unique: we could get the same phase by using a different set of matrices:
    \begin{align*}
        B = \cketbra{0}{+}\,,\quad  Q = 0\,,\quad T(x) = e^{ixQ}Xe^{-ixQ}\,,
    \end{align*}
    which trades bulk uniformity with constant boundary.

    As in Example~\ref{example: phase-D=1}, we could have simply described the unitary concretely through its actions on the $N$-particle states. However, it is desirable to find an explicit expression for $U_\theta$ in terms of only the field operators. We claim that the corresponding field unitary is a string operator of the form
    \begin{align}
        U_\theta = e^{-i\omega K}\,,\quad K =  \int_I\dd x\,x (-1)^{\Pi(x)}n(x)
    \end{align}
    where $K$ is the Hermitian generator and $\Pi(x)$ is as defined in Lemma~\ref{lemma: left-continuity-Heaviside}. It is instructive to check its action on the first few $N$-particle Fock states and it generalizes straightforwardly for arbitrary Fock states. 
    
    For $N=0$ we have $U_\theta\ket{\Omega} = \ket{\Omega}$, and for $N=1$ 
    \begin{align}
        n(x)\psi^\dagger(y)\ket{\Omega} = \delta(x-y)\psi^\dagger(y)\ket{\Omega}\,.
    \end{align}
    Thus given a one-particle state \mbox{$\ket{1_f} = \int\dd y\,f(y)\psi^\dagger(y)\ket{\Omega}$} where $f\in L^2(I)$, we have
    \begin{align}
        U_\theta\ket{1_f} 
        &= \int_I \dd y f(y)  e^{-i\omega \int_I\dd x\,x\cdot (-1)^{\Pi(x)}n(x)}\psi^\dagger(y)\ket{\Omega}\notag\\
        &=  \int_I \dd y f(y)  e^{-i\omega y (-1)^{\Pi(y)}}\psi^\dagger(y)\ket{\Omega} \notag \\
        &=  \int_I \dd y f(y)  e^{-i\omega y}\psi^\dagger(y)\ket{\Omega}\,.
    \end{align}
    In the last equality we used Lemma~\ref{lemma: left-continuity-Heaviside}.

    The first non-trivial check is $N=2$, where we need an alternating phase
    \begin{align*}
        e^{i\theta_2(x_1,x_2)} = e^{-i\omega(x_1-x_2)}\,.
    \end{align*}
    For two-particle Fock state with symmetric smearing $f(x_1,x_2) = f(x_2,x_1)$, we write
    \begin{align}
        \ket{2_f} &\propto \int \dd^2x\,f(x_1,x_2)\psi^\dagger(x_1)\psi^\dagger(x_2)\ket{\Omega} \,.
    \end{align}
    and applying $K$ we get
    \begin{align}
        &K\psi^\dagger(x_1)\psi^\dagger(x_2)\ket{\Omega} \notag\\ 
        &=  \int_I\dd x\,x (-1)^{\Pi(x)}n(x)\psi^\dagger(x_1)\psi^\dagger(x_2)\ket{\Omega} \notag\\
        &= \sum_{j=1}^2 \int_I\dd x\,x (-1)^{\Pi(x)}\delta(x-x_j)\psi^\dagger(x_1)\psi^\dagger(x_2)\ket{\Omega} \notag\\
        &= \sum_{j=1}^2 \,x_j (-1)^{\Pi(x_j)}\psi^\dagger(x_1)\psi^\dagger(x_2)\ket{\Omega}\,.
    \end{align}
    Using Lemma~\ref{lemma: left-continuity-Heaviside},
    \begin{align*}
        x_1(-1)^{\Pi(x_1)}\psi^\dagger(x_1)\psi^\dagger(x_2) &= x_1(-1)^{\Theta(x_2-x_1)}\psi^\dagger(x_1)\psi^\dagger(x_2) \\
        x_2(-1)^{\Pi(x_2)}\psi^\dagger(x_1)\psi^\dagger(x_2) &= x_2(-1)^{\Theta(x_1-x_2)}\psi^\dagger(x_1)\psi^\dagger(x_2)
    \end{align*}
    so that since $x_1<x_2$ we have
    \begin{align}
        U_\theta \ket{2_f} 
        &= \int_I\!\! D^2x\,e^{-i\omega (x_2-x_1)} f(x_1,x_2)  \psi^\dagger(x_1)\psi^\dagger(x_2)\ket{\Omega}\,,
    \end{align}
    thus producing the alternating phase as required. 
    
    This procedure generalizes to higher-particle sectors: more generally, for $N$-particle sector we have 
    \begin{align}
        K\prod_{j=1}^N\psi^\dagger(x_j)\ket{\Omega} 
        &= \sum_{j=1}^N \,x_j (-1)^{\Pi(x_j)}\prod_{l=1}^N\psi^\dagger(x_l)\ket{\Omega}
    \end{align}
    and using Lemma~\ref{lemma: left-continuity-Heaviside} with $\theta = \pi$, we can write this as
    \begin{align}
        &K\prod_{j=1}^N\psi^\dagger(x_j)\ket{\Omega} \notag\\
        &= \sum_{j=1}^N \,x_j (-1)^{\sum_{l>j}\Theta(x_l-x_j)}\prod_{l=1}^N\psi^\dagger(x_l)\ket{\Omega}\notag\\
        &= \sum_{j=1}^N \,x_j (-1)^{N-j}\prod_{l=1}^N\psi^\dagger(x_l)\ket{\Omega}\,,
    \end{align}
    which gives the required alternating phase \eqref{eq: D=2-phases}.
    \exampleqed 
     
\end{example}
\noindent The permutation-phase cMPU has the generic form
\begin{align}
    U_\theta = \exp\rr{i\int_I\dd x\,\mathcal{F}\left[\Pi(x)\right]n(x)}
\end{align}
for some non-trivial choice of functional $\mathcal{F}$. In Example~\ref{example: phase-D=2} we had $\mathcal{F}[\Pi(x)] = xe^{i \pi \Pi(x)}$. Thus the tensor-network structure manifests through the string operator $\Pi(x)$ in the generator of the unitary.

Although Proposition~\ref{proposition: generalized-permutation-cMPU} provides us with a large supply of phase cMPUs, not all phase cMPUs are permutation-phase cMPUs.
Using the path-ordered ansatz \eqref{eq: phase cMPU-exp}, it is possible to construct a different family of phase cMPUs.
\begin{proposition}[Number-controlled phase cMPU]
    \label{proposition: number-controlled-phase-cMPU}
    Let $J^-_D$ be the $D$-dimensional representation of the lowering operator of $\mathfrak{su}(2)$ with $D\geq 2$. Given the phase cMPU ansatz \eqref{eq: phase cMPU-exp}, the following family of matrices
    \begin{align*}
        B &= 1\oplus (-1+e^{i\theta})(J^-)^{D-1}\,,\\
        Q &= 0\,,\\
        T &= 1\oplus N_D\,, (N_D)_{ab} = \delta_{a,a+1} 
    \end{align*}
    produces a bulk-uniform phase cMPU with bond dimension $D+1$ up to some redundant gauge transformation. The resulting unitary is
    \begin{align}
        U_\theta = \openone + (e^{i\theta}-1)\delta_{D-2,j}\int D^j x \,\Ad_x^j(\Omega)
        \label{eq: number-controlled-phase-cMPU}
    \end{align}
    that adds a non-trivial phase on the $(D-2)$-particle sector. 
\end{proposition}
\begin{proof}
    This follows by direct computation: the key is to observe that the non-trivial phase arises because $N_D^{D-1} = (J^+)^{D-1}$ so it picks up a phase only when the matrix-product coefficient involves $D-1$ products of $T$'s. It is also possible to verify unitarity using Lemma~\ref{lemma: unitarity-v1}. 
\end{proof}
\noindent The number-controlled phase cMPU in Proposition~\ref{proposition: number-controlled-phase-cMPU} can be generalized to arbitrary number of phases on different particle number sectors via direct sums: for example, we can have
\begin{align*}
    B &= 1\oplus (-1+e^{i\theta_1})(J^-)^{D-1} \oplus (-1+e^{i\theta_2})(J^-)^{D'-1}\,,\\
    Q &= 0\,,\\
    T &= 1\oplus N_D\oplus N_{D'}\,, (N_D)_{ab} = \delta_{a,a+1} 
\end{align*}
so long as $D\neq D'\geq 2$, giving a phase cMPU with bond dimension $D+D'+1$. 

\begin{example}
    \label{example: projector-phases}
    Consider $D=3$ matrices
    \begin{align*}
        B = 1\oplus (-1+e^{i\theta})\sigma^-\,,\quad Q=0\,,\quad T = 1\oplus \sigma^+\,.
    \end{align*}
    Then the resulting operator $U_\theta$ is a diagonal cMPU 
    \begin{align}
        U_\theta = \openone + (e^{i\theta}-1)\int_I \dd x\,\psi^\dagger(x)\ketbra{\Omega}{\Omega}\psi(x)
        \label{eq: projector-phases}
    \end{align}
    that adds a nontrivial phase on 1-particle Fock states. 
    \exampleqed 
\end{example}
Observe that in this example, since $U_\theta$ is constructed out of phase cMPS, which is a continuum limit of phase MPS, it is possible to read off the relevant local tensors of its discrete counterpart:
\begin{align*}
    B = 1\oplus (-1+e^{i\theta})\sigma^-\,,\quad A^0  = \openone \,,\quad A^1 \approx  \sqrt{\epsilon }1\oplus \sigma^+
\end{align*}
where we recall that cMPS is based on having $A^{0}\sim \openone + \epsilon Q$ and $A^1 \propto L$. By removing the $\sqrt{\epsilon}$ scaling required for the continuum limit, this gives us the controlled phase unitary
\begin{align}
    U_{N,\theta} = \openone^{\otimes N} + (e^{i\theta}-1)\sum_{j=1}^N \sigma^+_j\ketbra{00...0}{00...0}\sigma^-_j
\end{align}
which is manifestly the discrete analog of $U_\theta$. This is a generic feature of cMPOs as a natural continuum limit of MPO: one should be able to obtain the discrete MPO whose limit is a given cMPO by reading off the bond matrices of the cMPO.

\subsection{Beyond phase cMPU}
\label{sec: cMPU-beyond-phases}

In this section we construct several families of cMPUs that are not diagonal in the particle number basis outside the phase unitary family.

\subsubsection{Displaced phase cMPU}

Using phase cMPU as a basis, we can now go beyond the diagonal cMPUs using the displacement unitaries. To do so, we first recall the following result from the discrete MPUs that is closely related to the $LU$-equivalence of $(2,2)$-LME states in Proposition~\ref{proposition: (2,2)-lme}.
\begin{proposition}[\cite{Styliaris2025matrixproduct}]
    Every $N$-qubit unitary $U$ admitting an $(2,2)$-LME state compression is locally unitary (LU) equivalent to a phase unitary. That is, given a phase unitary $U_\theta$ we have
    \begin{align}
        U \stackrel{LU}{=} \rr{\bigotimes_{j} V_j}U_\theta \rr{\bigotimes_{j} W_j}
    \end{align}
    where $V_j,W_j$ are local unitaries on each site $j$.
\end{proposition}
\noindent In the continuum, the role of local unitaries $V\coloneqq \bigotimes_{j} V_j$ is naturally taken by unitaries of the form
\begin{align}
    V\coloneqq e^{ i\int_I \dd x\,O(x)}
\end{align}
for some Hermitian $O(x)$, as is evident by discretizing $V$. The following example provides a natural candidate for the  cMPU analog of MPUs that are LU-equivalent to phase unitaries {mentioned in the main text}.
\begin{example}[Displaced phase cMPUs]
    \label{example: displaced-phase-cMPU}
    Let $U_\theta$ be a phase cMPU and $V,W$ be $D=1$ cMPUs of the form
    \begin{align*}
        V,W \sim e^{\int\dd x\,i r(x)n(x)+\alpha(x)\psi^\dagger(x) - \alpha(x)^*\psi(x) + \theta(x)\openone }
    \end{align*}
    that forms a subclass of non-squeezing local Gaussian unitaries. Then
    \begin{align}
        U = VU_\theta W 
    \end{align}
    is a cMPU with bond dimension $D$. More generally, the local unitaries $V,W$ can be decomposed into products of $D=1$ cMPUs of the form
    \begin{align}
        V= e^{i\int \dd x f_1(x)n(x)}e^{\int \dd x f_2(x)\psi^\dagger(x)-h.c.} e^{i\int \dd x\,f_3(x)\openone}
    \end{align}
    for some choice of $f_1,f_2,f_3$ that can be found by either the BCH formula, or by using non-Hermitian finite-dimensional representation of the Lie algebra generated by $n(x),\psi^\dagger(x),\psi(x),\openone$ (see, e.g., \cite{gilmore2008lie}). 
    
    The resulting unitary $U$ is not diagonal in the Fock basis of $\psi^\dagger(x)$ as $V,W$ do not generically commute with $\psi^\dagger(x)\psi(x)$ unless $V,W$ are both chosen to be phase unitaries in the same Fock basis.  To see this, let us suppose $V$ is a displacement with amplitude $\alpha(x)$ and $W=\openone$, and suppose $B_\theta,Q_\theta,T_\theta$ are the tensors for phase cMPU. According to the product relation \eqref{eq: product-of-MPO-tensors}, the resulting operator is still bond dimension $D$ cMPU but with the tensors
    \begin{align*}
        B &= B_\theta\,,\qquad\quad\quad   Q = Q_\theta -\frac{1}{2}|\alpha(x)|^2\openone_D \,,\\
        L &= \alpha(x)\openone_D \,,\qquad R = -\alpha(x)T^*_\theta\,,\\
        T &= T_\theta\,.
    \end{align*}
    As expected, the resulting cMPU has nonvanishing $L,R$ which shows that it is no longer a diagonal cMPU. 
    \exampleqed
\end{example}

\subsubsection{Finite-dimensional cMPU}

We say that a cMPU is finite-dimensional if it acts non-trivially only on a finite-dimensional subspace of the full bosonic Hilbert space. The first such example we had was from the number-controlled phase cMPU in Proposition~\ref{proposition: number-controlled-phase-cMPU}. We might expect that more examples should be possible because it acts non-trivially only on finite-dimensional subspaces, provided we work in Fock spaces where the Fock vacuum $\ket{\Omega}$ resides for all system sizes.

The upcoming example of finite-dimensional cMPU is based on a variant of the number-controlled phase cMPU, but it does not fit the ansatz \eqref{eq: phase cMPU-exp}.  We first prove the following lemma to simplify the construction.

\begin{lemma}
    \label{lemma: finite-projections}
    Let $\ket{\psi_i}$ be a finite collection of pairwise $m$ orthonormal states in some Hilbert space $\mathfrak{H}$ and $P_i = \ketbra{\psi_i}{\psi_i}$ be its rank-1 projector. Then
    \begin{align}
        U_{(m)} &\coloneqq  \openone  + \sum_{j=1}^m \rr{e^{ia_j}-1}P_j\,.
    \end{align}
    is unitary for any $a_j\in \R$. 
\end{lemma}
\begin{proof}
    Define a Hermitian operator $A=\sum_{j=1}^m a_j P_j$ for some $a_j\in \R$. Using orthonormality of $\ket{\psi_j}$'s, we have $A^n = \sum_j a_j^n P_j$, so that
    \begin{align*}
        e^{i A} &= \sum_{n=0}^\infty\frac{1}{n!}\sum_j a_j^nP_j = \openone + \sum_j (e^{i a_j}-1)P_j \equiv U_{(m)}\,,
    \end{align*}
    hence $U_{(m)}$ is unitary.
\end{proof}

The unitary $U_{(m)}$ in Lemma~\ref{lemma: finite-projections} is agnostic to the dimensionality of $\mathfrak{H}$ and does not rely on any tensor-network assumptions. Indeed, if $P_j$'s are chosen to be some highly entangled ``volume law'' states (not a (c)MPS), then it cannot be a (c)MPU as as it will fail to preserve the area-law entanglement. For example, the multi-control $Z$-gate in discrete MPU setting \cite{Styliaris2025matrixproduct} can be written as
\begin{align*}
    U = \openone^{\otimes N} - 2\ketbra{1}{1}^{\otimes N}
\end{align*}
which amounts to having $U_{(1)}$ with $a_1 = \pi$ and $P_1 = \ketbra{0}{0}^{\otimes N}$ and  is a projector over the span of a $D=1$ MPS $\ket{1}^{\otimes N}$. If we replace $\ketbra{1}{1}^{\otimes N}$ with an arbitrary projector $P_\Psi$ acting on $\mathfrak{H}$ then its action on some normalized MPS $\ket{\psi[B,A]}$
\begin{align*}
    (\openone - 2P_\Psi)\ket{\psi[B,A]} = \ket{\psi[B,A]} - 2\braket{\Psi|\psi[B,A]}\ket{\Psi}
\end{align*}
is not an MPS if $\ket{\Psi}$ itself is not an MPS. 

This argument generalizes to the continuum and gives us a family of cMPUs on Fock space. A natural question to ask is whether there are simple choices of $m$ pairwise-orthonormal cMPS that fulfills the above requirements, since two generic cMPS in Fock space will not be orthogonal. This is indeed possible, at least for the following subfamily of cMPS. 
\begin{lemma}[Single mode $N$-particle Fock states]
    \label{lemma: single-mode-Fock}
    Let $f\in L^2(I)$ be some square-integrable function with $||f||_2^2 = \int_I\dd x\,|f|^2 = 1$. Then the family of {normalized} $N$-particle Fock states in some mode $f$, namely
    \begin{align}
        \ket{N_f} = \frac{1}{\sqrt{N!}}\psi^\dagger(f)^N\ket{\Omega}\,,\quad \psi^\dagger(f) = \int_I\dd x\,f(x)\psi^\dagger(x)\,,
    \end{align}
    admits an exact cMPS representation with constant bond dimension $D = N+1$ where $N\in \mathbb{N}$. The local tensors are given by $Q(x) = 0$ and 
    \begin{align}
        L(x) = f(x)J^-\,,\qquad B = C_N\cdot (J^{+})^N\,,
    \end{align}
    where $J^\pm$ are the $D$-dimensional matrix representation of $\mathfrak{su}(2)$ ladder operators, and $C_N$ is some constant that depends on $N$ but is independent of the system size $\ell$. 
\end{lemma}

\begin{proof}
    This follows immediately from the fact that in the cMPS coefficients in Eq.~\eqref{eq: cMPS-Fock}, the trace over matrix products over the $j$-th particle sector reads
    \begin{align*}
        \Tr(B V_-^1L_1V_1^2...L_k V_k^+) &=  \Tr(B L^k)f(x_1)...f(x_k)\,.
    \end{align*}
    Using the standard notation for spin-$j$ representation of $\mathfrak{su}(2)$ that $J^\pm\ket{j,m} = \sqrt{j(j+1)-m(m\pm 1)}\ket{j,m\pm 1}$, we identify $D = 2j+1$ so $N=2j$. Observe that $(J^+)^{N}(J^-)^k$ is traceless unless $N=k$, in which case we have
    \begin{align*}
        (J^+)^N(J^-)^N = (N!)^2 \cketbra{\frac{N}{2},\frac{N}{2}}{\frac{N}{2},\frac{N}{2}}\,.
    \end{align*}
    so that $\Tr(BL^N)= C_N(N!)^2$. To normalize the state we use the fact that
    \begin{align*}
        &\int D^Nx\,\Tr(BL^N)f(x_1)...f(x_N)\psi^\dagger(x_1)...\psi^\dagger(x_N)\ket{\Omega} \notag\\
        &=\frac{1}{N!}\int \dd^Nx\,C_N(N!)^2f(x_1)...f(x_N)\psi^\dagger(x_1)...\psi^\dagger(x_N)\ket{\Omega}
    \end{align*}
    so to get the prefactor $1/\sqrt{N!}$ we set $C_N = (N!)^{-3/2}$.
\end{proof}

The following result shows that a rank-1 operator that maps one cMPS $\ket{\psi_1}$ to another cMPS $\ket{\psi_2}$ is itself a cMPO with bond dimension $D_1D_2$. This implies that the rank-1 projector over the subspace spanned by a cMPS $\ket{\psi}$ is a cMPO with bond dimension $D^2$.

\begin{proposition}
    \label{proposition: two-cMPS-projection}
    Let $\ket{\psi_1},\ket{\psi_2}$ be two normalized cMPS with the respective local tensors $B_j,Q_j,L_j$ with bond dimension $D_j$ where $j=1,2$. Then the operator $O_{ij}\coloneqq \ketbra{\psi_i}{\psi_j}$ is a cMPO with bond dimension $D=D_iD_j$. Furthermore, $O_{ij}$ admits an exponential ansatz \eqref{eq: cMPO-ansatz}
    \begin{align}
         \ketbra{\psi_i}{\psi_j} = \Tr(B_{ij}\mathcal{P}e^{\int\dd x\,\mathfrak{L}_{ij}(x)})[\ketbra{\Omega}{\Omega}]
    \end{align}
    where for each $i,j=1,2$ the local tensors and boundary matrices are
    \begin{equation}
    \begin{aligned}
        B_{ij}    &\coloneqq B^{\phantom{*}}_i\otimes B_j^*\,,\\
        Q_{ij}(x) &\coloneqq Q_i(x)\otimes\openone_{D_2} + \openone_{D_1}\otimes Q_j^*(x)\\
        L_{ij}(x) &\coloneqq L_i(x)\otimes\openone_{D_2} \,,\\
        R_{ij}(x) &\coloneqq \openone_{D_1}\otimes L^*_j(x)\,,\\
        T_{ij}(x) &= 0\,.
    \end{aligned}
    \end{equation}
\end{proposition}
\begin{proof}
    The proof follows by expanding each cMPS $\ket{\psi_i},\bra{\psi_j}$ in Dyson series (\textit{cf.} Eq.~\eqref{eq: cMPS-Fock}) and use the fact that $\Tr(X)\Tr(Y) = \Tr(X\otimes Y)$ and $\exp(X)\otimes \exp(Y) = \exp(X\otimes \openone + \openone\otimes Y)$. 
\end{proof}

\begin{proposition}[cMPU over cMPS subspaces]
    \label{proposition: cMPU over cMPS subspace}
    Consider a collection of $m$ rank-1 orthonormal projectors
    \begin{align}
        P_j\coloneqq \ketbra{\psi[B_j,Q_j,L_j]}{\psi[B_j,Q_j,L_j]}\,,
    \end{align}
    where $\{\ket{\psi[B_j,Q_j,L_j]}\}_{j=1}^m$ are $m$ pairwise orthonormal cMPS in $\mathfrak{F}(\mathcal{H})$. Then the unitary 
    \begin{align}
        U_{(m)}\coloneqq \openone + \sum_{j=1}^m\rr{e^{ia_j}-1}P_j
    \end{align}
    is a cMPU with bond dimension $D = 1 + \sum_{j=1}^m D_j^2$. {The cMPO ansatz is given by direct sum of the respective tensors}. 
\end{proposition}

\begin{proof}

    Since $U_{(m)}$ is a finite linear combinations of $m+1$ cMPOs with total bond dimension  at most $D \leq  \sum_{j=0}^m D_j$. At the same time, unitarity follows from Eq.~\eqref{lemma: finite-projections}, which is simpler to check than going through Lemma~\ref{lemma: unitarity-v1}.
\end{proof}

\begin{example}[Controlled-phase over one-particle subspace]
    \label{example: Fock-state-controlled-phase}
    Consider the unitary of the form
    \begin{align}
        U_\theta &= \openone + (e^{i\theta}-1)\ketbra{1_f}{1_f}\\
        &= \openone + (e^{i\theta}-1)\int\dd x\,\dd y\,f^*(x)f(y)\psi^\dagger(x)\ketbra{\Omega}{\Omega}\psi(y)\,.  \notag
    \end{align}
    This is the special case of Lemma~\ref{lemma: finite-projections} for $m=1$. The projector in the second term always takes 1-particle Fock state  $\psi^\dagger(g)\ket{\Omega}$ in some mode $g$ to another 1-particle Fock state $\psi^\dagger(f)\ket{\Omega}$, so it is not a mere multiplication by a phase. Consequently, it is not equivalent to the phase unitary family in Proposition~\ref{proposition: number-controlled-phase-cMPU}. Since by Proposition~\ref{proposition: two-cMPS-projection} the projector $\ketbra{1_f}{1_f}$ is a cMPO with bond dimension $D=4$, we know that $U_{\theta}$ is a $D=5$ cMPO and its unitarity follows from Lemma~\ref{lemma: finite-projections}, hence it is a cMPU. 
    \exampleqed
\end{example}

\begin{example}[Swapping vacuum and 1-particle subspace]
    \label{example: swapping vacuum}
    Consider the operator
    \begin{align}
        U = \openone + \ketbra{1_f}{\Omega} + \ketbra{\Omega}{1_f} - \ketbra{\Omega}{\Omega} - \ketbra{1_f}{1_f}\,.
    \end{align}
    where $||f||_2^2=1$. This operator is manifestly unitary and swaps vacuum and one-particle sector. Notice that this is a non-diagonal variant of Example~\ref{example: Fock-state-controlled-phase} by observing that
    \begin{align}
        U = \openone - 2\ketbra{\pm_f}{\pm_f}\,,\quad \ket{\pm_f} =\frac{1}{\sqrt{2}}\rr{\ket{\Omega}\pm\ket{1_f}}\,. 
    \end{align}
    Lemma~\ref{lemma: single-mode-Fock} shows that $\ket{1_f}$ is a cMPS with $D=2$, which implies that $\ket{\pm_f}$ is a cMPS with bond dimension $D=3$. Consequently, by Proposition~\ref{proposition: cMPU over cMPS subspace}, it is a cMPU with bond dimension $D=10$. 
    \exampleqed
\end{example}

\section{Generalized cMPO}

Below we provide some examples of ``generalized cMPO'', obtained by suitably modifying the cMPO ansatz \eqref{eq: cMPO-ansatz} directly. This amounts to relaxing condition (ii) that the cMPO arises as a continuum limit of some discrete MPO as prescribed in this work. 

The first is
\begin{align}
    O\coloneqq \Tr(B\mathcal{P}e^{\int\dd x\,Q(x)\otimes \Id[\cdot]  + \psi^\dag(x-a)[\cdot]\psi(x)})(\ketbra{\Omega}{\Omega})\,,
\end{align}
where $a\in I$. This corresponds to modifying the adjoint action $\Ad_x\equiv \psi^\dagger(x)[\cdot]\psi(x)$ to have translated field operator on the left. To illustrate its action, we see that for $D=1$ with $B=1,Q=0,T=1$, we have
\begin{align}
    O&= \mathcal{P}e^{\int\dd x\,\psi^\dag(x-a)[\cdot]\psi(x)})(\ketbra{\Omega}{\Omega}) \notag\\
    &\hspace{-0.3cm}= \sum_{j=0}^\infty\int D^jx\,\psi^\dagger(x_1-a)...\psi^\dagger(x_j-a)\ketbra{\Omega}{\Omega}\psi(x_1)...\psi(x_j) \,.
\end{align}
Its action each $j$-particle sector is essentially that of translation, since for any state (not necessarily a cMPS)
\begin{align*}
    \ket{\psi} = \sum_{j=0}^\infty\int D^jx\, c_j(x_1,...,x_j)\psi^\dagger(x_1)...\psi^\dagger(x_j)\ket{\Omega}
\end{align*}
we have
\begin{align*}
    &O\ket{\psi} \notag\\
    &= \sum_{j=0}^\infty\int D^jx\, c_j(x_1,...,x_j)\psi^\dagger(x_1-a)...\psi^\dagger(x_j-a)\ket{\Omega}\,.
\end{align*}
Consequently, $O$ is the unitary generating translation if we also impose either periodic boundary conditions or the field vanishes at infinity:
\begin{align}
    O = e^{-a\int\dd x\,\psi^\dagger(x)\partial_x\psi(x)} \equiv e^{-ia\partial_x}\,.
\end{align}

The second example is defined as
\begin{align*}
    O\coloneqq \lim_{\epsilon\to 0}\Tr(B\mathcal{P}e^{\int\dd x\,Q(x)\otimes \Id[\cdot]  + T(x)\otimes \mathcal{A}_\epsilon[\cdot]})(\ketbra{\Omega}{\Omega})\,,
\end{align*}
where
\begin{align*}
    \mathcal{A}_\epsilon\coloneqq \frac{i a}{\epsilon^2}\!\!\rr{\psi^\dagger(x+\epsilon)-\psi^\dagger(x-\epsilon)}[\cdot]\rr{\psi(x+\epsilon)-\psi(x-\epsilon)}
\end{align*}
and $a$ has units of length. This amounts to replacing $\Ad_x\equiv \psi^\dagger(x)[\cdot]\psi(x)$ in the cMPO ansatz with $\mathcal{A}_0\coloneqq ia\partial_x\psi^\dagger(x)[\cdot]\partial_x\psi(x)$.  Suppose we set $D=1,B=1,Q=0,T=1$. Then its action on the basis states $\ket{j}\coloneqq\psi^\dagger(x_1)...\psi^\dagger(x_j)\ket{\Omega}$ gives
\begin{align}
    O\ket{j} = \prod_{n=1}^j -a\partial_{x_j}^2\psi(x_j)\ket\Omega
\end{align}
using the symmetric difference for the approximation of the second derivative. Assuming that the boundary term vanishes under the integration by parts, the operator corresponds to
\begin{align}
    O \equiv e^{-i a \int \dd x\,\partial_x\psi^\dagger(x)\partial_x\psi(x)}
\end{align}
where the generator of the unitary $O$ is proportional to the kinetic part of the Lieb-Liniger Hamiltonian \cite{ganahl2017liebliniger}.

\end{document}